\documentclass[UKenglish,11pt,a4paper,twoside]{article}%
\usepackage{amsfonts}
\usepackage{amsmath}
\usepackage{amssymb}
\usepackage{graphicx}%
\providecommand{\U}[1]{\protect\rule{.1in}{.1in}}
\newtheorem{theorem}{Theorem}
\newtheorem{corollary}{Corollary}
\newtheorem{definition}{Definition}

\newtheorem{lemma}{Lemma}

\newenvironment{proof}[1][Proof]{\noindent\textbf{#1.} }{\ \rule{0.5em}{0.5em}}
\begin{document}

\title{When the bispectrum is real-valued}
\author{E. Igl\'{o}i and Gy. Terdik\\{\small Institute of Informatics, University of Debrecen}}
\date{}
\maketitle

\begin{abstract}
Let $\left\{  X(t),t\in\mathbb{Z}\right\}  $ be a stationary time series with
a.e.\ positive spectrum. Two consequences of that the bispectrum of $\left\{
X(t),t\in\mathbb{Z}\right\}  $ is real-valued but nonzero: 1) if $\left\{
X(t),t\in\mathbb{Z}\right\}  $ is also linear, then it is reversible; 2)
$\left\{  X(t),t\in\mathbb{Z}\right\}  $ can not be causal linear. A corollary
of the first statement: if $\left\{  X(t),t\in\mathbb{Z}\right\}  $ is linear,
and the skewness of $X(0)$ is nonzero, then third order reversibility implies
reversibility. In this paper the notion of bispectrum is of a broader scope.

\end{abstract}

\section{Introduction}

If a time series is reversible, then all of its polyspectra, if they exist,
are real-valued. The frequency-domain test of reversibility in \cite{Hinich}
uses this property regarded to the bispectrum, i.e.\ that real-valuedness of
the bispectrum is a necessary condition of reversibility. In this paper we
prove that, in essence, when the time series is linear, the real-valuedness of
the non-zero bispectrum is a sufficient condition as well, see Theorem
\ref{realbisp_reversi_theorem}. This confirms that when linearity is known to
hold, then for testing reversibility 1) there are no need for the polyspectra
of order higher than three, and 2) the bispectrum-based reversibility test of
\cite{Hinich} is consistent (with respect to non-reversibility, and not only
with respect to the alternative hypothesis that the bispectrum is not real).
There is also another corollary, valid in essence for linear time series with
a skewed distribution: third order reversibility (see Definition
\ref{kth_order_reversi_def}) implies reversibility, see Corollary
\ref{rev3_rev_corolla}.

Our other theorem, in essence: if the spectrum is positive and the bispectrum
is real-valued but nonzero, then the time series can not be causal linear, see
Theorem \ref{nonlinea_theorem}.

Let us recall some notions. A time series $\left\{  X(t),t\in\mathbb{Z}%
\right\}  $ is called \emph{reversible}, if%
\[
\left(  X_{t},X_{t+1},\ldots,X_{t+k}\right)  \overset{\text{d}}{=}\left(
X_{t+k},X_{t+k-1},\ldots,X_{t}\right)
\]
for all $k\in\mathbb{N}$ and $t\in\mathbb{Z}$ ($\,\overset{\text{d}}{=}$ means
equality in distribution). Reversibility implies stationarity, see
\cite{Lawrance}. A time series $\left\{  X(t),t\in\mathbb{Z}\right\}  $ is
reversible, if and only if it is stationary and%
\begin{equation}
\left(  X_{t_{1}},\ldots,X_{t_{k}}\right)  \overset{\text{d}}{=}\left(
X_{-t_{1}},\ldots,X_{-t_{k}}\right)  \label{revdefi}%
\end{equation}
for all $k\in\mathbb{N}$ and $t_{1}<\ldots<t_{k}\in\mathbb{Z}$. Since Gaussian
stationary time series are always reversible, it is enough to deal with the
non-Gaussian case.

A time series $\left\{  X(t),t\in\mathbb{Z}\right\}  $ is called
\emph{linear}, if it has a moving average representation%
\begin{align}
&  X(t)=\sum\limits_{k=-\infty}^{\infty}c(k)Z(t-k),\label{MA_rep}\\
&  \sum\limits_{k=-\infty}^{\infty}c(k)^{2}<\infty,\;\text{r.v.s }%
Z(t),t\in\mathbb{Z}\text{, are i.i.d. with }\operatorname*{E}Z(t)=0,\text{
}\operatorname*{E}Z(t)^{2}<\infty.\nonumber
\end{align}
Obviously, linearity implies stationarity. Because of non-Gaussianity, the
representation (\ref{MA_rep}) is unique, apart from constant multiplier and
time shift, see \cite{Cheng1} or \cite{Rosenblatt}, Theorem 1.3.1. A linear
representation of the type (\ref{MA_rep}) is called \emph{causal}, if the
summation is over nonnegative indices only, i.e.\ if the time series does not
depend on future $Z(t)$ values.

The proofs of the main results depend largely on the solution of a particular
case of the Cauchy functional equation, see Lemma \ref{Cauchy_lemma}.

The rest of the paper is organized as follows. In Section \ref{bisp_sect} the
notion of the bispectrum is generalized in order to be existing for, among
others, linear time series with finite third order absolute moment. The main
results of the paper are stated in Section \ref{main_sect}. The proofs and the
necessary lemmas on the Cauchy functional equation are presented in Section
\ref{proof_sect}.

\section{The bispectrum of a linear time series\label{bisp_sect}}

Assume that the time series $\left\{  X(t),t\in\mathbb{Z}\right\}  $ is
stationary in third order. Let us denote the joint cumulant of the random
variables $X(t_{1}),X(t_{2}),X(t_{3})$ by $\operatorname*{cum}\left(
X(t_{1}),X(t_{2}),X(t_{3})\right)  $. Because of stationarity we have
$\operatorname*{cum}\left(  X(t_{1}),X(t_{2}),X(t_{3})\right)
=\operatorname*{cum}\left(  X(0),X(t_{2}-t_{1}),X(t_{3}-t_{1})\right)  $,
$t_{1},t_{2},t_{3}\in\mathbb{Z}$, thus the third order joint cumulant is, in
fact, a function of two variables only. The bispectrum has been defined in
\cite{Bril_polyspect-65} as the two variable Fourier-transform of the sequence
of third order joint cumulants $\operatorname*{cum}\left(  X(0),X(t_{1}%
),X(t_{2})\right)  $, $t_{1},t_{2}\in\mathbb{Z}$. For the Fourier-transform to
be meaningful, it has been required that the cumulant series be absolutely
summable. Defined in this way, the bispectrum is an integrable function, thus
$\operatorname*{cum}\left(  X(0),X(t_{1}),X(t_{2})\right)  $, $t_{1},t_{2}%
\in\mathbb{Z}$, is the two variable inverse Fourier-transform of it. The
absolute summability condition is, however, too strict, e.g.\ long range
dependent time series generally fail to fulfil it. However, if we define the
bispectrum requiring the integrability of the bispectrum only, but not the
absolute summability of the cumulant series, then we get a more general
concept, what is extensive enough to apply to at least any linear time series
with finite absolute moments of third order. On the other hand, the
integrability of a function guarantees the one-to-one correspondence between
itself and its inverse Fourier transform. Thus, the above mentioned classical
definition of the bispectrum can be generalized so that we do not assume
absolutely summable cumulants.

\begin{definition}
\label{bisp_def}Let $\left\{  X(t),t\in\mathbb{Z}\right\}  $ be a stationary
time series with finite third order absolute moment, and assume that a
function $B(\omega_{1},\omega_{2})$ defined a.e.\ on $[0,2\pi)\times
\lbrack0,2\pi)$ is integrable, and its inverse Fourier transform is just the
cumulant sequence, i.e.%
\begin{equation}
\operatorname*{cum}\left(  X(0),X(t_{1}),X(t_{2})\right)  =\int\limits_{0}%
^{2\pi}\int\limits_{0}^{2\pi}\exp\left(  i\left(  t_{1}\omega_{1}+t_{2}%
\omega_{2}\right)  \right)  B(\omega_{1},\omega_{2})d\omega_{1}d\omega_{2},
\label{cum_bisp_correspondence}%
\end{equation}
$t_{1},t_{2}\in\mathbb{Z}$. Then function $B(\omega_{1},\omega_{2})$ is called
the bispectrum of $\left\{  X(t),t\in\mathbb{Z}\right\}  $.
\end{definition}

In the rest of the paper we use the notion of bispectrum in the sense of
Definition \ref{bisp_def}. As we have already mentioned, there is a one-to-one
correspondence between the set of all bispectra and the set of those third
order joint cumulant sequences for which the bispectrum exists.

\begin{lemma}
\label{triple}Let $\left\{  X(t),t\in\mathbb{Z}\right\}  $ be a linear time
series with finite third order absolute moment and moving average
representation (\ref{MA_rep}). Then the bispectrum of $\left\{  X(t),t\in
\mathbb{Z}\right\}  $ exists, and it has the form%
\begin{equation}
B(\omega_{1},\omega_{2})=\frac{\operatorname*{cum}\nolimits_{3}(Z(0))}{\left(
2\pi\right)  ^{2}}\varphi(\omega_{1})\varphi(\omega_{2})\varphi(-\omega
_{1}-\omega_{2}) \label{B_omega}%
\end{equation}
for a.e.\ $\left(  \omega_{1},\omega_{2}\right)  \in\lbrack0,2\pi
)\times\lbrack0,2\pi)$, where
\[
\varphi(\omega)=\sum\limits_{k=-\infty}^{\infty}c(k)e^{-ik\omega},
\]
$\omega\in\lbrack0,2\pi)$, is the frequency domain transfer function
corresponding to the linear representation (\ref{MA_rep}).
\end{lemma}

\section{Some new relations among reversibility, bispectrum and
linearity\label{main_sect}}

\begin{theorem}
\label{realbisp_reversi_theorem}Let $\left\{  X(t),t\in\mathbb{Z}\right\}  $
be a linear time series with finite third order absolute moment and
a.e.\ positive spectrum. If its bispectrum is real-valued but not a.e.\ zero,
then $\left\{  X(t),t\in\mathbb{Z}\right\}  $ is reversible.
\end{theorem}

Sometimes the 2 and 3 dimensional distributions can be handled more directly
then the general finite dimensional ones. Motivated by this, we introduce the
following notion, and then state a corollary of the previous theorem.

\begin{definition}
\label{kth_order_reversi_def}Let $k\in\left\{  2,3,\ldots\right\}  $. A
stationary time series $\left\{  X(t),t\in\mathbb{Z}\right\}  $ is reversible
in $k^{\text{th}}$ order, if (\ref{revdefi}) holds for all $t_{1}<\ldots
<t_{k}\in\mathbb{Z}$.
\end{definition}

\begin{corollary}
\label{rev3_rev_corolla}Let $\left\{  X(t),t\in\mathbb{Z}\right\}  $ be a
linear time series with finite third order absolute moment, and a.e.\ positive
spectrum. If the skewness of $X(t)$ is nonzero, then $\left\{  X(t),t\in
\mathbb{Z}\right\}  $ is reversible in third order, if and only if it is reversible.
\end{corollary}

The following theorem is about a relation between real-valuedness of the
bispectrum and causal linear representability.

\begin{theorem}
\label{nonlinea_theorem}Let $\left\{  X(t),t\in\mathbb{Z}\right\}  $ be a time
series with a.e.\ positive spectrum. If the bispectrum of $X(t)$ exists and is
real-valued but not a.e.\ zero, then $\left\{  X(t),t\in\mathbb{Z}\right\}  $
can not have a causal linear representation.
\end{theorem}

\section{Proofs and preliminary lemmas\label{proof_sect}}

\begin{proof}
[Proof of Lemma \ref{triple}]First of all by the Marcinkiewicz--Zygmund
inequality r.v.s $Z(t)$ also have finite third order moments. Applying the
Schwarz--Cauchy inequality and utilizing the $2\pi$-periodicity of the
function $\varphi(\omega)$, one can easily get that $B(\omega_{1},\omega_{2})$
in (\ref{B_omega}) is integrable, i.e.\ $B\in L^{1}\left(  [0,2\pi
)\times\lbrack0,2\pi)\right)  $. Let us introduce the notation%
\[
\varphi^{\left(  K\right)  }(\omega)\doteq\sum\limits_{k=-K}^{K}%
c(k)e^{-ik\omega},
\]
$K\in\mathbb{N}$, $\omega\in\lbrack0,2\pi)$. In the same manner as the
integrability of $B(\omega_{1},\omega_{2})$, one can also get the inequality%
\[
\int\limits_{0}^{2\pi}\int\limits_{0}^{2\pi}\left\vert \psi_{1}(\omega
_{1})\psi_{2}(\omega_{2})\psi_{3}(-\omega_{1}-\omega_{2})\right\vert
d\omega_{1}d\omega_{2}\leqslant\sqrt{2\pi}\left\Vert \psi_{1}\right\Vert
_{2}\left\Vert \psi_{2}\right\Vert _{2}\left\Vert \psi_{3}\right\Vert _{2}%
\]
for any $2\pi$-periodic functions $\psi_{1},\psi_{2},\psi_{3}$ with the
property $\psi_{i}|_{[0,2\pi)}\in L^{2}[0,2\pi),\;i=1,2,3$, where $\left\Vert
\cdot\right\Vert _{2}$ is the $L^{2}[0,2\pi)$-norm. Thus we have%
\[
\underset{K_{1},K_{2},K_{3}\rightarrow\infty}{\text{l.i.m.}}\left(
\frac{\operatorname*{cum}\nolimits_{3}(Z(0))}{\left(  2\pi\right)  ^{2}%
}\varphi^{(K_{1})}(\omega_{1})\varphi^{(K_{2})}(\omega_{2})\varphi^{(K_{3}%
)}(-\omega_{1}-\omega_{2})\right)  =B(\omega_{1},\omega_{2}),
\]
where l.i.m. denotes limit in $L^{1}\left(  [0,2\pi)\times\lbrack
0,2\pi)\right)  $. Hence the Fourier-transform of $B(\omega_{1},\omega_{2})$
is%
\begin{align}
&  \int\limits_{0}^{2\pi}\int\limits_{0}^{2\pi}e^{i\left(  t_{1}\omega
_{1}+t_{2}\omega_{2}\right)  }B(\omega_{1},\omega_{2})d\omega_{1}d\omega
_{2}\label{egyik}\\[0.11in]
&  =\frac{\operatorname*{cum}\nolimits_{3}(Z(0))}{\left(  2\pi\right)  ^{2}%
}\underset{K_{1},K_{2},K_{3}\rightarrow\infty}{\lim}\int\limits_{0}^{2\pi}%
\int\limits_{0}^{2\pi}e^{i\left(  t_{1}\omega_{1}+t_{2}\omega_{2}\right)
}\varphi^{(K_{1})}(\omega_{1})\varphi^{(K_{2})}(\omega_{2})\varphi^{(K_{3}%
)}(-\omega_{1}-\omega_{2})d\omega_{1}d\omega_{2}\nonumber\\[0.11in]
&  =\frac{\operatorname*{cum}\nolimits_{3}(Z(0))}{\left(  2\pi\right)  ^{2}%
}\underset{K_{1},K_{2},K_{3}\rightarrow\infty}{\lim}\sum\limits_{k_{1}=-K_{1}%
}^{K_{1}}\sum\limits_{k_{2}=-K_{2}}^{K_{2}}\sum\limits_{k_{3}=-K_{3}}^{K_{3}%
}c(k_{1})c(k_{2})c(k_{3})\nonumber\\[0.11in]
&  \qquad\qquad\qquad\qquad\qquad\qquad\qquad\qquad\qquad\quad\times
\int\limits_{0}^{2\pi}\!e^{i\omega_{1}\left(  t_{1}-k_{1}+k_{3}\right)
}d\omega_{1}\int\limits_{0}^{2\pi}\!e^{i\omega_{2}\left(  t_{2}-k_{2}%
+k_{3}\right)  }d\omega_{2}\nonumber\\
&  =\operatorname*{cum}\nolimits_{3}(Z(0))\sum\limits_{k=-\infty}^{\infty
}c(t_{1}+k)c(t_{2}+k)c(k).\nonumber
\end{align}
On the other hand we have%
\begin{align}
\operatorname*{cum}  &  \left(  X(0),X(t_{1}),X(t_{2})\right)  \label{masik}%
\\[0.11in]
&  =\operatorname*{cum}\left(  \sum\limits_{k_{0}=-\infty}^{\infty}%
c(k_{0})Z(-k_{0}),\sum\limits_{k_{1}=-\infty}^{\infty}c(k_{1})Z(t_{1}%
-k_{1}),\sum\limits_{k_{2}=-\infty}^{\infty}c(k_{2})Z(t_{2}-k_{2})\right)
\nonumber\\[0.11in]
&  =\sum\limits_{k_{0}=-\infty}^{\infty}\sum\limits_{k_{1}=-\infty}^{\infty
}\sum\limits_{k_{2}=-\infty}^{\infty}c(k_{0})c(k_{1})c(k_{2}%
)\operatorname*{cum}\left(  Z(-k_{0}),Z(t_{1}-k_{1}),Z(t_{2}-k_{2})\right)
\nonumber\\[0.11in]
&  =\operatorname*{cum}\nolimits_{3}(Z(0))\sum\limits_{k=-\infty}^{\infty
}c(k)c(t_{1}+k)c(t_{2}+k).\nonumber
\end{align}
From (\ref{egyik}) and (\ref{masik}) we have%
\[
\operatorname*{cum}\left(  X(0),X(t_{1}),X(t_{2})\right)  =\int\limits_{0}%
^{2\pi}\int\limits_{0}^{2\pi}e^{i\left(  t_{1}\omega_{1}+t_{2}\omega
_{2}\right)  }B(\omega_{1},\omega_{2})d\omega_{1}d\omega_{2},
\]
meaning that function $B(\omega_{1},\omega_{2})$ in (\ref{B_omega}) is the
bispectrum of $\left\{  X(t),t\in\mathbb{Z}\right\}  $.
\end{proof}

As a preliminary to the proof of Theorem \ref{realbisp_reversi_theorem} we
solve the Cauchy functional equation modulo $\pi$, defined a.e.. It has been
solved, separately, both when it is defined a.e. and when it holds modulo
$\pi$. At first we quote these results, both in simplified form.

\begin{lemma}
\label{deBruijn_lemma}(\cite{Jurkat}, \cite{deBruijn}) Let $\left(
Y,+\right)  $ be a commutative group, and let the function $f:\mathbb{R}%
\rightarrow Y$ satisfy the Cauchy functional equation a.e., i.e.%
\begin{equation}
f(x+y)=f(x)+f(y) \label{f_additive}%
\end{equation}
for almost all pairs $(x,y)\in\mathbb{R}^{2}$ (in the sense of Lebesgue
measure on $\mathbb{R}^{2}$). Then there exists a function $g:\mathbb{R}%
\rightarrow Y$ satisfying (\ref{f_additive}) everywhere in $\mathbb{R}^{2}$
and being a.e.\ (in the sense of Lebesgue measure on $\mathbb{R}$) equal to f.
\end{lemma}

\begin{lemma}
\label{Baron_Kann_lemma}(\cite{BaronKannappan}) Let $F$ be a real topological
vector space and assume that $L:F\rightarrow\mathbb{R}$ is a continuous linear
functional. Suppose $\varphi:\mathbb{R}\rightarrow F$ satisfies%
\[
\varphi(x+y)-\varphi(x)-\varphi(y)\in L^{-1}\left(  \mathbb{Z}\right)
\]
for all $x,y\in\mathbb{R}$. If $\varphi$\ is measurable, then there exists a
continuous linear operator $M:\mathbb{R}\rightarrow F$ such that%
\[
\varphi(x)-M(x)\in L^{-1}\left(  \mathbb{Z}\right)
\]
for all $x\in\mathbb{R}$.
\end{lemma}

Now, consider the Cauchy functional equation when it holds modulo $\pi$ and
a.e., simultaneously.

\begin{lemma}
\label{Cauchy_lemma}Let the measurable function $f:\mathbb{R\rightarrow R}$
fulfil the congruence%
\[
f(x+y)=\left(  f(x)+f(y)\right)  \,\operatorname{mod}\,\pi
\]
for a.e.\ $(x,y)\in\mathbb{R}^{2}$, i.e.,%
\[
f(x+y)-f(x)-f(y)\in\pi\mathbb{Z}%
\]
for a.e.\ $(x,y)\in\mathbb{R}^{2}$. Then $f$ must be of the form%
\[
f(x)=cx+k(x)
\]
for a.e.\ $x\in\mathbb{R}$, where $c\in\mathbb{R}$ and $k:\mathbb{R\rightarrow
}\pi\mathbb{Z}$.
\end{lemma}

\begin{proof}
First we prove that there exists a measurable function $h:\mathbb{R}%
\rightarrow\mathbb{R}$, such that%
\begin{equation}
h(x)=f(x)\,\operatorname{mod}\,\pi\label{h_eq_f}%
\end{equation}
for a.e.\ $x\in\mathbb{R}$, and%
\[
h(x+y)=\left(  h(x)+h(y)\right)  \,\operatorname{mod}\,\pi
\]
for all $(x,y)\in\mathbb{R}^{2}$. (Notice the difference between
\textquotedblleft for a.e.\ $(x,y)\in\mathbb{R}^{2}$ \textquotedblright\ and
\textquotedblleft for all $(x,y)\in\mathbb{R}^{2}$ \textquotedblright.)
Consider the factor group $\left(  \mathbb{R}/\left(  \pi\mathbb{Z}\right)
,\oplus\right)  =\left(  [0,\pi),\oplus\right)  $, where $\oplus$ is the
modulo $\pi$ addition. This is a commutative group. There exists a measurable
function $g:\mathbb{R}\rightarrow\lbrack0,\pi)$, such that%
\begin{equation}
g(x)=f(x)\,\operatorname{mod}\,\pi\label{g_equ_f}%
\end{equation}
for all $x\in\mathbb{R}$, and%
\[
g(x+y)=\left(  g(x)+g(y)\right)  \,\operatorname{mod}\,\pi,
\]
i.e.%
\begin{equation}
g(x+y)=g(x)\oplus g(y) \label{g_egyenlet}%
\end{equation}
for a.e.\ $(x,y)\in\mathbb{R}^{2}.$ (To see this take the function $g$ to be
$f\,\operatorname{mod}\,\pi$.) Thus by Lemma \ref{deBruijn_lemma} there exists
a function $h:\mathbb{R}\rightarrow\lbrack0,\pi)$, such that%
\begin{equation}
h(x)=g(x) \label{h_eq_g}%
\end{equation}
for a.e.\ $x\in\mathbb{R}$, and%
\begin{equation}
h(x+y)=h(x)\oplus h(y) \label{h_egyenlet}%
\end{equation}
for all $(x,y)\in\mathbb{R}^{2}$. Considering now that (\ref{g_equ_f}) holds
for all $x\in\mathbb{R}$, and (\ref{h_eq_g}) holds for a.e.\ $x\in\mathbb{R}$,
it follows that (\ref{h_eq_f}) holds for a.e.\ $x\in\mathbb{R}$.

Next we prove that the function $h$ must be of the form%
\begin{equation}
h(x)=cx+\ell(x) \label{h_keplete}%
\end{equation}
for all $x\in\mathbb{R}$, where $c\in\mathbb{R}$ and $\ell:\mathbb{R}%
\rightarrow\pi\mathbb{Z}$. The conditions of Lemma \ref{Baron_Kann_lemma} are
fulfilled, since:\newline$\bullet$ by (\ref{h_egyenlet}) we have%
\[
h(x+y)-h(x)-h(y)\in\pi\mathbb{Z}%
\]
\hphantom{$\bullet$\ } for all $x,y\in\mathbb{R}$;\newline$\bullet$ $h$ is
measurable, since so is $g$, and (\ref{h_eq_g}) holds.\newline Thus by Lemma
\ref{Baron_Kann_lemma} there exist a constant $c\in\mathbb{R}$ and a function
$\ell:\mathbb{R}\rightarrow\pi\mathbb{Z}$ such that (\ref{h_keplete}) holds
for all $x\in\mathbb{R}$.

Combining the results of the previous two parts we get the statement of the lemma.
\end{proof}

\bigskip

\begin{proof}
[Proof of Theorem \ref{realbisp_reversi_theorem}]By Lemma \ref{triple}
$\left\{  X(t),t\in\mathbb{Z}\right\}  $ has a bispectrum $B(\omega_{1}%
,\omega_{2})$ of the form (\ref{B_omega}). Thus, using the notations of Lemma
\ref{triple}, we have%
\begin{equation}
\varphi(\omega_{1})\varphi(\omega_{2})\overline{\varphi(\omega_{1}+\omega
_{2})}=\frac{4\pi^{2}}{\operatorname*{cum}\nolimits_{3}(Z(0))}B(\omega
_{1},\omega_{2}) \label{alapegyenl}%
\end{equation}
for a.e.\ $(\omega_{1},\omega_{2})\in\lbrack0,2\pi)\times\lbrack0,2\pi)$,
where $\operatorname*{cum}\nolimits_{3}(Z(0))\neq0$ because the bispectrum is
not identically zero. Since the spectrum of $\left\{  X(t),t\in\mathbb{Z}%
\right\}  $, denoted by $S(\omega)$, is a.e.\ positive, and $S(\omega
)=\left\vert \varphi(\omega)\right\vert ^{2}$ for a.e.\ $\omega\in
\lbrack0,2\pi)$, it follows that $\varphi(\omega)\neq0$ for a.e.\ $\omega
\in\lbrack0,2\pi)$. Hence $B(\omega_{1},\omega_{2})\neq0$ for a.e.\ $(\omega
_{1},\omega_{2})\in\lbrack0,2\pi)\times\lbrack0,2\pi)$ (this can be seen by
dividing the set where the bispectrum is zero into three sets corresponding to
the three factors on the left hand side of (\ref{alapegyenl}), and observing
that each of these three sets is of zero measure). Taking the logarithm and
then the imaginary part in (\ref{alapegyenl}), we have%
\begin{equation}
\psi\left(  \omega_{1}\right)  +\psi\left(  \omega_{2}\right)  -\psi\left(
\omega_{1}+\omega_{2}\right)  \in\pi\mathbb{Z}, \label{psziegyenl}%
\end{equation}
for a.e.\ $(\omega_{1},\omega_{2})\in\lbrack0,2\pi)\times\lbrack0,2\pi)$,
where%
\[
\psi\left(  \omega\right)  \doteq\operatorname{Im}\log\left(  \varphi
(\omega)\right)  ,
\]
and the principal branch $\log:\mathbb{C}\setminus\left\{  0\right\}
\rightarrow\mathbb{R}+i[-\pi,\pi)$ of the complex logarithm function is used.

Denote the periodic continuation of $\psi$ by the same letter, i.e.%
\begin{gather*}
\psi:\mathbb{R\rightarrow R},\\
\psi(\omega+2\pi)=\psi(\omega),
\end{gather*}
for all $\omega\in\mathbb{R}$. Thus we have (\ref{psziegyenl}) for
a.e.\ $(\omega_{1},\omega_{2})\in\mathbb{R}^{2}$. Now, by Lemma
\ref{Cauchy_lemma} we have%
\begin{equation}
\psi(\omega)=c\omega+k(\omega) \label{psi_eq(c)}%
\end{equation}
for a.e.\ $\omega\in\mathbb{R}$, where $c\in\mathbb{R}$ and
$k:\mathbb{R\rightarrow}\pi\mathbb{Z}$. Using the $2\pi$-periodicity of $\psi$
it follows from (\ref{psi_eq(c)}) that%
\[
c=\frac{n}{2},
\]
for some $n\in\mathbb{Z}$. Thus we have%
\[
\psi(\omega)=\frac{n}{2}\omega+k(\omega)
\]
for a.e.\ $\omega\in\mathbb{R}$, where $n\in\mathbb{Z}$ and
$k:\mathbb{R\rightarrow}\pi\mathbb{Z}$. Substituting this form of $\psi
(\omega)$ into the argument of the transfer function $\varphi(\omega)$ we have%
\begin{equation}
\varphi(\omega)=r(\omega)e^{i\psi(\omega)}=r(\omega)e^{i\frac{n}{2}\omega
}e^{ik(\omega)}, \label{fi_trigonom}%
\end{equation}
where $r(\omega)=\left\vert \varphi(\omega)\right\vert $.

Let us calculate the coefficients in the moving average representation
(\ref{MA_rep}) of $X(t)$. We have%
\begin{gather}
c(\ell)=\frac{1}{2\pi}\int\limits_{-\pi}^{\pi}e^{-i\ell\omega}\varphi
(\omega)d\omega=\frac{1}{2\pi}\int\limits_{-\pi}^{\pi}r(\omega)e^{ik(\omega
)}e^{i\left(  \frac{n}{2}-\ell\right)  \omega}d\omega\label{c_sym}\\
=\frac{1}{2\pi}\int\limits_{-\pi}^{\pi}r(\omega)e^{ik(\omega)}e^{i\left(
\ell-\frac{n}{2}\right)  \omega}d\omega=\frac{1}{2\pi}\int\limits_{-\pi}^{\pi
}\varphi(\omega)e^{-i(n-\ell)\omega}d\omega=c(n-\ell),\nonumber
\end{gather}
for each $\ell\in\mathbb{Z}$, where in the second equation we used
(\ref{fi_trigonom}), while the third equation follows from the fact that both
the coefficient $c(\ell)$ and $e^{ik(\omega)}$ are real. Now, (\ref{c_sym})
means that the sequence of coefficients in (\ref{MA_rep}) is symmetric to the
index $n$. Hence, using the necessary and sufficient condition of
reversibility of \cite{Cheng2} (what states, that a linear time series with
a.e.\ positive spectrum is reversible if and only if either the series
$\left\langle c\right\rangle $ in (\ref{MA_rep}) is symmetric to some index,
or it is skew-symmetric and the r.v.\ $Z(0)$ has symmetric distribution) it
follows the statement of our theorem.
\end{proof}

\begin{proof}
[Proof of Corollary \ref{rev3_rev_corolla}]If $\left\{  X(t),t\in
\mathbb{Z}\right\}  $ is reversible in third order, then it is reversible also
in second order. Thus for $k=3$ relation (\ref{revdefi}) holds even if the
indices are not all different. Thus%
\[
\operatorname*{cum}\left(  X(t_{1}),X(t_{2}),X(t_{3})\right)
=\operatorname*{cum}\left(  X(-t_{1}),X(-t_{2}),X(-t_{3})\right)
\]
for all $t_{1},t_{2},t_{3}\in\mathbb{Z}$, particularly%
\[
\operatorname*{cum}\left(  X(0),X(t_{1}),X(t_{2})\right)  =\operatorname*{cum}%
\left(  X(0),X(-t_{1}),X(-t_{2})\right)
\]
for all $t_{1},t_{2}\in\mathbb{Z}$. Hence, using
(\ref{cum_bisp_correspondence}) and the one-to-one correspondence between the
bispectra and the cumulants, it follows that the bispectrum is real-valued.
Moreover, from (\ref{masik}) we have $\operatorname*{cum}\nolimits_{3}%
(X(0))=\,$\textit{constant}$\,\times\operatorname*{cum}\nolimits_{3}(Z(0))$,
implying that $\operatorname*{cum}\nolimits_{3}(Z(0))\neq0$, since
$\operatorname*{cum}\nolimits_{3}(X(0))=\operatorname*{E}\left(  X(0)\right)
^{3}\neq0$. On the other hand, by By Lemma \ref{triple} $\left\{
X(t),t\in\mathbb{Z}\right\}  $ has a bispectrum $B(\omega_{1},\omega_{2})$ of
the form (\ref{B_omega}). $B(\omega_{1},\omega_{2})$ is not a.e.\ zero,
because otherwise either $\operatorname*{cum}\nolimits_{3}(Z(0))=0$ or the
transfer function $\varphi(\omega)$ would be zero on a set of positive
Lebesgue measure, and then the spectrum $S(\omega)=\left\vert \varphi
(\omega)\right\vert ^{2}$ would also be zero on a set of positive measure,
what is a contradiction. Hence by Theorem \ref{realbisp_reversi_theorem}
follows the reversibility of $\left\{  X(t),t\in\mathbb{Z}\right\}  $.
\end{proof}

\bigskip

\begin{proof}
[Proof of Theorem \ref{nonlinea_theorem}]Let us assume linearity and repeat
the proof of Theorem \ref{realbisp_reversi_theorem} up to the conclusion that
the sequence of coefficients in the linear representation is symmetric. We are
ready, because a linear representation with symmetric coefficients are
necessarily two-sided, while a causal representation would be
one-sided.\bigskip
\end{proof}

\textbf{Acknowledgment. }The publication was supported by the
T\'{A}MOP-4.2.2.C-11/1/KONV-2012-0001 project. The project has been supported
by the European Union, co-financed by the European Social Fund.

\begin{proof}
\bigskip
\end{proof}

\end{document}